\documentclass{article}

\usepackage[utf8]{inputenc}
\usepackage{amsmath}
\usepackage{amsfonts}
\usepackage{amssymb}
\usepackage{amsthm}

\usepackage{nccmath}
\usepackage{makecell}
\usepackage{braket}
\usepackage{bussproofs}

\usepackage[capitalise]{cleveref}
\usepackage{soul} 

\usepackage{graphicx}
\usepackage{multicol}
\usepackage{footmisc}

\usepackage{url}

\usepackage{tikz}
\usetikzlibrary{positioning}

\newcommand{\CS}{\mathsf{CS}}

\newcommand{\Tm}{\mathsf{Tm}}
\newcommand{\atm}{\mathsf{atm}}

\newcommand{\Prop}{\mathsf{Prop}}
\newcommand{\ATm}{\mathsf{ATm}}

\newcommand{\JUP}{\mathsf{JUS}}
\newcommand{\JUPcs}{\mathsf{JUS_{CS}}}
\newcommand{\LJUP}{\mathsf{L_{JUS}}}
\newcommand{\Taut}{\mathsf{Taut}}
\newcommand{\App}{\mathsf{App}}
\newcommand{\Pers}{\mathsf{Pers}}

\newcommand{\Up}{\mathsf{Up}}

\newcommand{\Indep}{\mathsf{Indep}}
\newcommand{\Funct}{\mathsf{Funct}}
\newcommand{\Norm}{\mathsf{Norm}}

\newcommand{\up}{\mathsf{up}}

\newcommand{\model}{\mathcal{M}}
\newcommand{\WMP}{W_{MP}}

\newcommand{\MC}{\mathcal{M}^C}
\newcommand{\rk}{\ell}

\makeatletter
\newcommand{\xLongLeftRightArrow}[2][]{\ext@arrow0055{\LongLeftRightArrowfill@}{#1}{#2}}
\def\LongLeftRightArrowfill@{\arrowfill@\Leftarrow\Relbar\Rightarrow}
\makeatother
\makeatletter
\def\th@plain{%
	\thm@notefont{}
	\itshape 
}

\theoremstyle{plain}
\newtheorem{theorem}{Theorem}
\newtheorem{lemma}[theorem]{Lemma}

\newtheorem{remark}[theorem]{Remark}

\theoremstyle{definition}
\newtheorem{definition}[theorem]{Definition}

\begin{document}

\title{Belief Expansion in Subset Models}
\date{}
\author{Eveline Lehmann \and Thomas Studer}

\maketitle

\begin{abstract}
Subset models provide a new semantics for justifcation logic. The main idea of subset models is that evidence terms are interpreted as sets of possible worlds. A term then justifies a formula if that formula is true in each world of the interpretation of the term.

In this paper, we introduce a belief expansion operator for subset models. We study the main properties of the resulting logic as well as the differences to a previous (symbolic) approach to belief expansion in justification logic.
%
\end{abstract}

\section{Introduction}

Justification logic is  a  variant of  modal logic where the $\Box$-modality is replaced with a familiy of so-called  evidence terms, i.e.~instead of formulas $\Box F$, justification logic features formulas of the form $t:F$ meaning \emph{$F$ is known for reason~$t$}~\cite{ArtFit11SEP,artemovFittingBook,justificationLogic}.

The first justification logic, the Logic of Proofs, has been developed by Artemov~\cite{Artemov1995OperationalModalLogic,Artemov2001ExplicitProvability} in order to provide intuitionistic logic with a classical provability semantics. Thus evidence terms represent proofs in a formal system like Peano arithmetic. By \emph{proof} we mean a Hilbert-style proof, that is a sequence of formulas 
\begin{equation}\label{eq:sequence:1}
F_1,\ldots,F_n 
\end{equation}
where each formula is either an axiom or follows by a rule application from formulas that occur earlier in the sequence.
A justification formula $t:F$ holds in this arithmetical semantics if $F$ occurs in the proof represented by $t$. Observe that $F$ need not be the last formula in the sequence~\eqref{eq:sequence:1}, but can be any formula $F_i$ in it, i.e.~we think of proofs as multi-conclusion proofs~\cite{Artemov2001ExplicitProvability,KSweak}.

After the Logic of Proofs has been introduced, it was observed that terms can not only represent mathematical proofs but evidence in general. Using this interpretation, justification logic provides a versatile framework for epistemic logic~\cite{Art06TCS,Art08RSL,Baltag201449,BucKuzStu11JANCL,BucKuzStu14Realizing,komaogst}.
In order to obtain a semantics of evidence terms that fits this general reading, one has to ignore the order of the sequence~\eqref{eq:sequence:1}. That is evidence terms are interpreted simply as sets of formulas.

This is anticipated in both Mkrtychev models~\cite{Mkr97LFCS} as well as Fitting models~\cite{Fit05APAL}.
The former are used to obtain a decision procedure for justification logic where one of the main steps is to keep track of which (set of) formulas a term justifies, see, e.g.,~\cite{justificationLogic,Stu13JSL}.  The latter provide first epistemic models for justification logic where each possible world is equipped with an evidence function that specifies which terms serve as possible evidence for which (set of) formulas in that world.

Artemov~\cite{Art12SLnonote} conceptually addresses the problem of the logical type of justifications. 
He claims that in the logical setting, justifications are naturally interpreted as sets of formulas. He introduces so-called modular models, which are  based on the basic interpretation of justifications as sets of propositions and the convenience assumption of
\begin{equation}\label{eq:jyb:1}
\text{justification yields belief.}\tag{JYB}
\end{equation}
That means if a term justifies a formula (i.e., the formula belongs to the interpretation of the term), then that formula is believed (i.e., true in all accessible possible worlds)~\cite{KuzStu12AiML}. 
Note that \eqref{eq:jyb:1} has been dropped again in more recent versions of modular models~\cite{artemovFittingBook}.

So let us consider models for justification logic that interpret terms as sets of formulas.
A belief change operator on such a model will operate by changing those sets of formulas (or introducing new sets, etc.). Dynamic epistemic justification logics have been studied, e.g., in~\cite{BucKuzStu11WoLLIC,BucKuzStu14Realizing,KuzStu13LFCS,Ren12Synthesenonote}.
Kuznets and Studer~\cite{KuzStu13LFCS}, in particular,  introduce a justification logic with an operation for belief expansion. Their system satisfies a Ramsey principle as well as minimal change. In fact, their system meets all AGM postulates for belief expansion.

In their model, the belief expansion operation is monotone: belief sets can only get larger, i.e., 
\begin{equation}\label{eq:bigger:1}
\text{belief expansion always only adds new beliefs.} 
\end{equation}
This is fine for first-order beliefs. Indeed, one of the AGM postulates for expansion requires that beliefs are persistent.
However, as we will argue later, this behavior is problematic for higher-order beliefs.

In this paper, we present an alternative approach that behaves better with respect to higher-order beliefs. It uses subset models for justification logics. This is a recently introduced semantics \cite{StuderLehmannSubsetModel2019,ExploringSM} that interprets terms not as sets of formulas but as sets of possible  worlds. There, a formula $t:A$ is true if the interpretation of $t$ is a subset of the truth-set of $A$, i.e., $A$ is true in each world of the interpretation of $t$.
Intuitively, we can read $t:A$ as $A$ is believed and $t$ justifies this belief.
Subset models lead to new operations on terms (like  intersection). Moreover, they provide a natural framework for probabilistic evidence (since the interpretation of a term is a set of possible worlds, we can easily measure it). Hence they support aggregation of probabilistic evidence~\cite{artemov2016onAggregatingPE,StuderLehmannSubsetModel2019}. They also naturally contain non-normal worlds and support paraconsistent reasoning.

It is the aim of this paper to equip subset models with an operation for belief expansion similar to~\cite{KuzStu13LFCS}. The main idea is to introduce justification terms $\up(A)$ such that after a belief expansion with $A$, we have that $A$ is believed and $\up(A)$ (representing the expansion operation on the level of terms) justifies this belief. Semantically, the expansion $A$ is dealt with by intersecting the interpretation of $\up(A)$ with the truth-set of $A$.  
This provides a better approach to belief expansion than~\cite{KuzStu13LFCS} as~\eqref{eq:bigger:1} will hold  for first-order beliefs but it will fail  in general.

The paper is organized as follows. 
In the next section we introduce the language and a deductive system for $\JUP$, a justification logic with belief expansion and subset models.
Then we present its semantics and establish soundness of $\JUP$.
Section~\ref{sec:persistence:1} is concerned with persistence properties of first-order and higher-order beliefs. Further we prove a Ramsey property for  $\JUP$.
Finally, we conclude the paper and mention some further work.

\subsubsection{Acknowledgements.} 
This work was supported by the Swiss National Science Foundation grant \emph{Explicit Reasons}, 200020\_184625.

\section{Syntax}

Given a set of countably many constants $c_i$, countably many variables $x_i$, and countably many atomic propositions $P_i$, terms and formulas of the language of $\JUP$ are defined as follows:
\begin{itemize}
	\item Evidence terms
	\begin{itemize}
		\item Every constant $c_i$ and every variable $x_i$ is an atomic term. If $A$ is a formula, then $\up(A)$ is an atomic term. Every atomic term is a term.
		\item If $s$ and $t$ are terms and $A$ is a formula, then $s\cdot_A t$ is a term.
	\end{itemize}
	\item Formulas
	\begin{itemize}
		\item Every atomic proposition $P_i$ is a formula.
		\item If $A$, $B$, $C$ are formulas, and $t$ is a term, then $\lnot A$, $A\to B$, $t:A$ and $[C]A$ are formulas.
	\end{itemize}
\end{itemize}

The annotation of the application operator may seem a bit odd at first. However, it is often used in dynamic epistemic justification logics, see, e.g.~\cite{KuzStu13LFCS,Ren12Synthesenonote}.

The set of atomic terms is denoted by $\ATm$, the set of all terms is denoted by $\Tm$. The set of atomic propositions is denoted by $\Prop$ and the set of all formulas is denoted by $\LJUP$.
We define the remaining classical connectives, $\bot$, $\land$, $\lor$, and $\leftrightarrow$, as usual making use of the law of double negation and de Morgan's laws.

The intended meaning of the justification term $\up(A)$ is that after an update with $A$, this act of updating serves as justification to believe $A$. Consequently, the justification term $\up(A)$ has no specific meaning before the update with $A$ happens.

\begin{definition}[Set of Atomic Subterms]\label{subterm}
	The set of atomic subterms of a term or formula is inductively defined as follows:
	\begin{itemize}
		\item $\atm(t):=\{t\} $ \quad if $t$ is a constant or a variable
		\item $\atm(\up(C)):=\{\up(C)\} \cup \atm(C)$ 
		\item $\atm(s\cdot_A t):=\atm(s)\cup\atm(t)\cup\atm(A)$
		\item $\atm(P):=\emptyset$ \quad for $P\in\Prop$
		\item $\atm(\lnot A):=\atm(A)$
		\item $\atm(A\to B):=\atm(A)\cup\atm(B)$
		\item $\atm(t:A):=\atm(t)\cup\atm(A)$
		\item $\atm([C]A):=\atm(A) \cup \atm(C)$. 
	\end{itemize}
\end{definition}

\begin{definition}
We call a formula $A$ \emph{up-independent} if for each subformula $[C]B$ of $A$ we have that $\up(C) \notin \atm(B)$.
\end{definition}

Using Definition \ref{subterm}, we can control that updates and justifications are independent. This is of importance to distinguish cases where updates change the meaning of justifications and corresponding formulas from cases where the update does not affect the meaning of a formula.

We will use the following notation: $\tau$ denotes a finite sequence of formulas and $\epsilon$ denotes the empty sequence. Given a sequence $\tau = C_1, \ldots C_n$ and a formula $A$, the formula $[\tau]A$ is defined by
\[
[\tau]A = [C_1]\ldots[C_n]A \text{ if $n > 0$} \quad\text{and}\quad [\epsilon]A := A.
\]

The logic $\JUP$ has the following axioms and rules where $\tau$ is a finite (possibly empty) sequence of formulas:
\begin{align}
1.\quad& [\tau]A \quad\text{for all propositional tautologies $A$}&\tag{$\Taut$}\\
2.\quad& [\tau](t:(A\to B)\land s:A\leftrightarrow t\cdot_A s:B)&\tag{$\App$}\\
3. \quad& [\tau]([C]A\leftrightarrow A) \quad\text{if $[C]A$ is up-independent}&\tag{$\Indep$}\\ 
4.\quad& [\tau]([C]\neg A\leftrightarrow\neg[C]A)&\tag{$\Funct$}\\
5.\quad& [\tau]([C](A\to B)\leftrightarrow([C]A\to[C]B))&\tag{$\Norm$}\\
6.\quad& [\tau][A]\up(A):A&\tag{$\Up$}\\
7.\quad& [\tau](\up(A):B \to [A]\up(A):B)&\tag{$\Pers$}
\end{align}

A constant specification $\CS$ for $\JUP$ is any subset
\[
\begin{split}
\CS\subseteq\{(c, \ [\tau_1]c_1 :  [\tau_2]c_2 :& \ldots : [\tau_n]c_n:  A) \enspace| \\
	&n\geq 0,  \text{ $c, c_1, \ldots, c_n$ are constants},\\
	& \text{$\tau_1, \ldots, \tau_n$ are sequences of formulas},\\
	& \text{$A$ is an axiom of $\JUP$}\}
\end{split}
\]

$\JUPcs$ denotes the logic $\JUP$ with the constant specification $\CS$. 
The rules of $\JUPcs$ are Modus Ponens and Axiom Necessitation:
\[
	\AxiomC{$A$}
	\AxiomC{$A\to B$}
	\RightLabel{(MP)}
	\BinaryInfC{$B$}
	\DisplayProof
\qquad\qquad
	\AxiomC{}
	\RightLabel{(AN)\quad if $(c, A)\in\CS$}
	\UnaryInfC{$[\tau]c:A$}
\DisplayProof
\]

Before establishing some basic properties of $\JUPcs$, let us briefly discuss its axioms.
The direction from left to right in axiom $(\App)$ provides an internalization of modus ponens. 
Because of the annotated application operator, we also have the other direction, which is a minimality condition. It states that a justification represented by a complex term can only come from an application of modus ponens.

Axiom $(\Indep)$ roughly states that an update with a formula $C$ can only affect the truth of formulas that contain certain update terms.

Axiom $(\Funct)$ formalizes that updates are functional, i.e.~the result of an update is uniquely determined.

Axiom $(\Norm)$, together with Lemma~\ref{l:upNec:1}, states that $[C]$ is a normal modal operator for each formula $C$.

Axiom $(\Up)$ states that after a belief expansion with $A$, the formula~$A$ is indeed believed and $\up(A)$ justifies that belief.

Axiom $(\Pers)$ is a simple persistency property of update terms.

\begin{definition}\label{def:cs_axiomatocally_appropriate}
A constant specification $\CS$ is called \emph{axiomatically appropriate} if 
\begin{enumerate}
\item
for each axiom $A$, there is a constant $c$ with $(c,A) \in \CS$ and
\item 
for any formula $A$ and any constant $c$, if $(c,A) \in \CS$, then for each sequence of formulas $\tau$ there exists a constant $d$ with  
\[
(d,[\tau]c:A) \in \CS.
\]
\end{enumerate}
\end{definition}

The first clause in the previous definition is the usual condition for an axiomatically appropriate constant specification (when the language includes the !-operation). Here we also need the second clause in order to have the following two lemmas, which establish that necessitation is admissible in $\JUPcs$. Both are proved by induction on the length of derivations.

\begin{lemma}\label{l:upNec:1}
Let $\CS$ be an arbitrary constant specification.
For all formulas $A$ and $C$ we have that if $A$ is provable in $\JUPcs$, then  $[C]A$ is provable in $\JUPcs$.
\end{lemma}

\begin{lemma}[Constructive Necessitation]
Let $\CS$ be an axiomatically appriopriate constant specification.
For all formulas $A$ we have that if $A$ is provable in $\JUPcs$, then there exists a term $t$ such that $t:A$ is provable in $\JUPcs$.
\end{lemma}

We will also need the following auxiliary lemma.
\begin{lemma}\label{l:aux:1}
Let $\CS$ be an arbitrary constant specification. For all terms $s,t$ and all formulas $A,B,C$, $\JUPcs$ proves:
\[
[C]t:(A\to B) \land [C]s:A \ \leftrightarrow\  [C]t\cdot_A s:B
\]
\end{lemma}


\section{Semantics}

Now we are going to introduce subset models for the logic  $\JUPcs$. In order to define a valuation function on these models, we will need the following  measure for the length of formulas.
\begin{definition}[Length]
The \emph{length} of a term or formula is inductively defined by:
\begin{align*}
&\rk(t):=1 \text{ if $t \in \ATm$}  && \rk(s \cdot_A t):=\rk(s)+\rk(t)+\rk(A)+1 \\
&\rk(P):=1 \text{ if $P \in \Prop$}  && \rk(A \to B):=\rk(A)+\rk(B)+1  \\
&\rk(\lnot A):=\rk(A)+1  && \rk(t:A):=\rk(t)+\rk(A)+1 \\
&\rk([B]A):= \rk(B)+\rk(A)+1 \hspace{-2em}
\end{align*}

\end{definition}

\begin{definition}[Subset Model] \label{def:subset_model} 
We define a \emph{subset model} \[\mathcal{M}=(W, W_0, V_1, V_0, E)\] for $\JUP$ by:
	\begin{itemize}
		\item $W$ is a set of objects called worlds.
		\item $W_0\subseteq W$, $W_0\neq\emptyset$.
		\item $V_1: (W \setminus W_0) \times\LJUP\to \{0, 1\}$.
		\item $V_0:W_0 \times\Prop\to \{0, 1\}$.
		\item $E: W\times\Tm \to \mathcal{P}(W)$ such that for  $\omega\in W_0$ and all  $A\in\LJUP$:
		\[
		E(\omega, s\cdot_A t)\subseteq E(\omega, s)\cap E(\omega, t)\cap\WMP,
		\]
		where $\WMP$ is the set of all deductively closed worlds, formally given by
	\begin{align*}
	\WMP&:= W_0  \cup \WMP^1 \quad\text{where}\\
	\WMP^1&:=\{\omega\in W\setminus W_0\ | \   \\
			&\forall A, B\in\LJUP\ ((V_1(\omega, A)=1\text{ and }
	V_1(\omega, A\to B)=1) \\
	&\qquad\qquad\qquad\quad\text{implies } V_1(\omega, B)=1)\}.
	\end{align*}		
	\end{itemize}	
\end{definition}
We call $W_0$ the set of \emph{normal} worlds. The worlds in $W \setminus W_0$ are called \emph{non-normal} worlds. 
$\WMP$ denotes the set of worlds where the valuation function (see the following definition) is closed under modus ponens.

In normal worlds, the laws of classical logic hold, whereas non-normal worlds may behave arbitrarily. In a non-normal world we may have that both $P$ and $\lnot P$  hold  or we may have that neither~$P$ nor $\lnot P$ holds.
We need non-normal worlds to take care of the hyperintensional aspects of justification logic. In particular, we must be able to model that constants do not justify all axioms. In normal worlds, all axioms hold. Thus we need non-normal worlds to make axioms false.

Let $\mathcal{M}=(W, W_0, V_1, V_0, E)$ be a subset model.
We define the \emph{valuation function} $V_\mathcal{M}$ for $\mathcal{M}$ and the \emph{updated model} $\mathcal{M}^C$ for any formula $C$ simultaneously. For $V_\mathcal{M}$, we  often drop the subscript $\mathcal{M}$ if it is clear from the context.

We define $V : W \times \LJUP\to \{0, 1\}$ as follows by induction on the length of formulas:
\begin{enumerate}
\item Case $\omega \in W \setminus W_0$. We set $V(\omega,F) := V_1(\omega,F)$;
\item Case $\omega \in W_0$. We define $V$ inductively by:
\begin{enumerate}
\item $V(\omega,P) := V_0(\omega,P)$ for $P \in \Prop$;
\item $V(\omega,t:F) := 1$ if{f}   $E(\omega,t) \subseteq \{\upsilon \in W \ |\ V(\omega,F) =1\}$ for $t \in \ATm$;
\item $V(\omega, s\cdot_F r:G)=1$ if{f} 
	$V(\omega, s:(F\to G))=1$ and $V(\omega, r:F)=1$; 
\item $V(\omega, \lnot F)=1$  if{f} $V(\omega, F)=0$;
\item $V(\omega, F \to G)=1$  if{f} $V(\omega, F)=0$ or $V(\omega, G)=1$;
\item $V(\omega, [C]F)=1$  if{f} $V_{\mathcal{M}^C}(\omega, F)=1$ where $V_{\mathcal{M}^C}$ is the valuation function for the updated model $\mathcal{M}^C$.
\end{enumerate}
\end{enumerate}
The following notation for the truth set of $F$ will be convenient:
\[
[[F]]_\model := \{\upsilon \in W \ |\ V_\mathcal{M}(\upsilon, F) = 1\}.
\]

The updated model $\MC = (W^{\MC}, W_0^{\MC}, V_1^{\MC}, V_0^{\MC}, E^{\MC})$ is given by:
\[
W^{\MC} := W \qquad W_0^{\MC}:=W_0 \qquad   V_1^{\MC} := V_1 \qquad V_0^{\MC} := V_0
\]
and
\[
 E^{\MC}(\omega,t) := \begin{cases}
 							E^{\model}(\omega,t) \cap [[C]]_{\MC} & \text{if $\omega \in W_0$ and $t=\up(C)$}\\
 							E^{\model}(\omega,t) & \text{otherwise}
 						\end{cases}		
\]

The valuation function for complex terms is well-defined.
\begin{lemma}\label{app_is-subset} 
	For a subset model $\model$ with a world $\omega\in W_0$, $s, t\in \Tm$, $A, B\in\LJUP$, we find that
	\[
	V(\omega, s\cdot_A t:B)=1\quad\text{  implies }\quad E(\omega, s\cdot_A t)\subseteq[[B]]_\model.\]
\end{lemma}
\begin{proof}
	The proof is by induction on the structure of $s$ and $t$:
	\begin{itemize}
		\item base case $s,t\in\ATm$:\\
		Suppose $V(\omega, s\cdot_A t:B)=1$. Case 2c of the definition of $V$ in Definition~\ref{def:subset_model} for normal worlds yields that \[V(\omega, s:(A\to B))=1 \text{ and }V(\omega, t:A)=1.\] With case 2b from the same definition we obtain \[E(\omega, s)\subseteq[[A\to B]]_\model\text{ and }E(\omega,t)\subseteq[[A]]_\model.\] Furthermore the definition of $E$ for normal worlds guarantees that 
		\[E(\omega, s\cdot_A t)\subseteq E(\omega, s)\cap E(\omega, t)\cap \WMP.\]
		So for each $\upsilon\in E(\omega, s\cdot_A t)$ there is $V(\upsilon, A\to B)=1$ and $V(\upsilon, A)=1$ and $\upsilon\in\WMP$ and hence either by the definition of $\WMP^1$ or by case 2e of the definition of $V$ in normal worlds there is $V(\upsilon, B)=1$. Therefore $E(\omega, s\cdot_A t)\subseteq[[B]]_\model$.
		\item $s, t\in \Tm$ but at least one of them is not atomic: w.l.o.g.~suppose $s=r\cdot_C q$. Suppose $V(\omega, s\cdot_A t:B)=1$ then $V(\omega, s:(A\to B))=1$ and $V(\omega, t:A)=1$. 
		Since  $s=r\cdot_C q$ and $\omega\in W_0$ we obtain \[V(\omega, r:(C\to(A\to B)))=1\text{ and }V(\omega, q:C)=1\] and by I.H.~that
		\[		
		E(\omega, r)\subseteq[[C\to(A\to B)]]_\model \text{ and } E(\omega, q)\subseteq[[C]]_\model.
		\]
		With the same reasoning as in the base case we obtain 
		\[
		E(\omega, s)=E(\omega, r\cdot_C q)\subseteq[[A\to B]]_\model.
		\]
		If $t$ is neither atomic, the argumentation works analoguously and since we have then shown both $E(\omega, s)\subseteq[[A\to B]]_\model$ and $E(\omega, t)\subseteq[[A]]_\model$, the conclusion is the same as in the base case.
		\qedhere
	\end{itemize}
\end{proof}

\begin{remark}
The opposite direction to Lemma \ref{app_is-subset} need not hold. 
Consider a model $\model$  and a formula $s\cdot_A t:B$ with atomic terms $s$ and~$t$ such that  \[V_\model(\omega, s\cdot_A t:B)=1\] and thus also $E(\omega, s\cdot_A t)\subseteq[[B]]_\model$.
Now consider a model $\model'$ which is defined like $\model$ except that
\[
E'(\omega, s) := E(\omega, t) \quad\text{and}\quad E'(\omega, t) := E(\omega, s).
\]
We observe the following:
\begin{enumerate}
\item
We have $E'(\omega, s\cdot_A t) = E(\omega, s\cdot_A t)$ as the condition 
\[
E'(\omega, s\cdot_A t)\subseteq E'(\omega, s)\cap E'(\omega, t)\cap \WMP
\]
still holds since intersection of sets is commutative.
Therefore  \[E'(\omega, s\cdot_A t)\subseteq[[B]]_{\model'}\]  holds.
\item
However, it need not be the case that
\[
E'(\omega, s)\subseteq[[A\to B]]_{\model'} \text { and } E'(\omega, t)\subseteq[[A]]_{\model'}.
\]
Therefore $V_{\model'}(\omega, s:(A \to B))=1$ and $V_{\model'}(\omega, t:A)=1$ need not hold and thus also
$V_{\model'}(\omega, s\cdot_A t:B)=1$ need not be the case anymore.
\end{enumerate}
\end{remark}

\begin{definition}[$\CS$-Model]
Let $\CS$ be a constant specification.
A subset model $\model = (W, W_0, V_1, V_0, E)$ is called a $\CS$-subset model or a subset model for $\JUPcs$ if for all $\omega \in W_0$ and for all $(c,A) \in \CS$ we have 
\[
E(\omega,c) \subseteq [[A]]_\model.
\]
\end{definition}

We observe that updates respect $\CS$-subset models.
\begin{lemma}
Let $\CS$ be an arbitrary constant specification and let $\model$ be a 
$\CS$-subset model. We find that $\MC$ is a $\CS$-subset model for any formula~$C$.
\end{lemma}

\section{Soundness}

\begin{definition}[Truth in Subset Models]\label{Fcstar-truth} Let 
\[
\mathcal{M}= (W, W_0, V_1, V_0, E)
\]
be a subset model, $\omega\in W$, and $F\in\LJUP$. We define the relation $\Vdash$ as follows:
	\[\mathcal{M}, \omega\Vdash F\quad\text{ if{f} }\quad V_\model(\omega, F)=1.\]
\end{definition}

\begin{theorem}[Soundness] 
 Let $\CS$ be an arbitrary constant specification. Let 
$
\mathcal{M}= (W, W_0, V_1, V_0, E)
$
be a $\CS$-subset model and $\omega\in W_0$.
For each formula $F\in\LJUP$ we have that
\[ 
\JUPcs\vdash F\quad\text{implies} \quad \mathcal{M}, \omega\Vdash F.
\]
\end{theorem}
\begin{proof}
As usual by induction on the length of the derivation of~$F$.
We only show the case where $F$ is an instance of axiom $(\Indep)$.

By induction on $[C]A$ we show that for all $\omega$
\[
\MC, \omega \Vdash A \qquad\text{if{f}} \qquad \model, \omega \Vdash A.
\]
We distinguish the following cases.
\begin{enumerate}
\item $A$ is an atomic proposition. Trivial. 
\item $A$ is $\lnot B$. By I.H.
\item $A$ is $B \to D$. By I.H.
\item $A$ is $t:B$. Subinduction on $t$:
	\begin{enumerate}
	\item $t$ is a variable or a constant. Easy using I.H. for $B$.
	\item $t$ is a term $\up(D)$. By assumption, we have that $C \neq D$. Hence this case is similar to the previous case.
	\item $t$ is a term $r \cdot_D s$. We know that $t:B$ is equivalent to 
	\[
	r: (D \to B) \land s:D.
	\]
	 Using I.H. twice, we find that
	\[
	\MC, \omega \Vdash r: (D \to B) \quad\text{and}\quad \MC, \omega \Vdash s:D
	\] 
	 if and only if
	\[
	\model, \omega \Vdash r: (D \to B) \quad\text{and} \quad \model, \omega \Vdash s:D.
	\] 
	 Now the claim follows immediately.
	\end{enumerate}
\item $A$ is $[D]B$. Making use of the fact that $A$ is up-independent, this case also follows using I.H.\qedhere
\end{enumerate}
\end{proof}

\section{Basic Properties}\label{sec:persistence:1}

We first show that first-order beliefs are persistent in $\JUP$.
Let $F$ be a formula that does not contain any justification operator. We have that 
if $t$ is a justification for $F$, then, after any update, this will still be the case.
Formally, we have the following lemma.

\begin{lemma}
For any term $t$ and any formulas $A$ and $C$ we have that
if\/ $A$ does not contain a subformula of the form $s:B$, then 
\[
t:A \ \to\ [C]t:A
\]
is provable.
\end{lemma}
\begin{proof}
We proceed by induction on the complexity of $t$ and distinguish the following cases:
\begin{enumerate}
\item Case $t$ is atomic and $t \neq \up(C)$. Since $A$ does not contain any evidence terms, the claim follows immediately from axiom ($\Indep$). 
\item Case $t = \up(C)$. This case is an instance of axiom ($\Pers$).
\item Case $t= r \cdot_B s$. From $r \cdot_B s : A$ we get by ($\App$)
\[
s:B \quad\text{and}\quad r:(B \to A).
\]
By I.H.~we find
\[
[C]s:B \quad\text{and}\quad [C]r:(B \to A).
\]
Using Lemma~\ref{l:aux:1} we conclude $[C] r \cdot_B s: A$.
\qedhere
\end{enumerate}
\end{proof}

Let us now investigate higher-order beliefs.
We argue that persistence should not hold in this context. 
Consider the following scenario.
Suppose that you are in a room together with other people. Further suppose that no announcement has been made in that room. Therefore, it is not the case that $P$ is believed because of an announcement. 
Formally, this is expressed by 
\begin{equation}\label{eq:notPers:1}
\lnot \up(P):P.
\end{equation}
We find that
\begin{equation}\label{eq:notPers:2}
\text{the fact that you are in that room} 
\end{equation}
justifies your belief in~\eqref{eq:notPers:1}.
Let the term $r$ represent~\eqref{eq:notPers:2}. Then we have
\begin{equation}\label{eq:notPers:3}
r : \lnot \up(P):P.
\end{equation}
Now suppose that $P$ is publicly announced in that room.
Thus we have in the updated situation
\begin{equation}\label{eq:notPers:4}
\up(P):P.
\end{equation}
Moreover, the fact that you are in that room justifies now your belief in~\eqref{eq:notPers:4}. Thus we have 
$
	r:\up(P):P
$
and hence in the original situation we have
\begin{equation}\label{eq:notPers:5}
[P]r : \up(P):P
\end{equation}%
and \eqref{eq:notPers:3} does no longer hold after the announcement of $P$.

The following lemma formally states that
persistence fails for higher-oder beliefs.

\begin{lemma}\label{l:higher:1}
There exist formulas $ r: B$ and $A$ such that
\[
r:B \to [A] r: B
\]
 is not provable.
\end{lemma}
\begin{proof}
Let $B$ be the formula $\neg\up(P):P$ and
consider the subset model \[\mathcal{M}=(W, W_0, V_1, V_0, E)\] with 
\[
W := \{\omega, \upsilon\} \quad 
W_0 := \{\omega\}\quad
V_1(\upsilon, P)=0\quad 
V_0(\omega, P)=1
\]
and
\[
E(\omega,r)=\{\omega\}\quad 
E(\omega,\up(P))=\{\omega,\upsilon\}.
\]
Hence $[[P]]_\model = \{\omega\}$ and thus $E(\omega,\up(P)) \not\subseteq [[P]]_\model$.
Since $\omega \in W_0$, this yields $V(\omega, \up(P):P)=0$.
Again by $\omega \in W_0$, this implies \[V(\omega, \lnot\up(P):P)=1.\]
Therefore $E(\omega,r) \subseteq [[\lnot\up(P):P]]_\model$ and using $\omega \in W_0$, we get
\[\model,\omega \Vdash r:\lnot\up(P):P\].

Now consider the updated model $\model^P$.
We find that 
\begin{equation*}
E^{\model^P}(\omega,\up(P)) = \{\omega\}
\end{equation*}
and thus 
$E^{\model^P}((\omega,\up(P))) \subseteq [[P]]_{\model^P}$.
Further, using $\omega \in W_0^{\model^P}$ we get \[V_{\model^P}(\up(P):P)=1\] and thus
$V_{\model^P}(\lnot \up(P):P)=0$.
That is 
$\omega \notin [[(\lnot \up(P):P]]_{\model^P}$.
We have $E^{\model^P}(\omega,r)=\{\omega\}$ and, therefore, 
$E^{\model^P}(\omega,r) \not\subseteq [[(\lnot \up(P):P]]_{\model^P}$.

With  $\omega \in W_0^{\model^P}$ we get
$\model^P, \omega \not\Vdash r:\lnot\up(P):P$.
We conclude \[\model, \omega \not\Vdash [P]r:\lnot\up(P):P.\]
\end{proof}

Next, we show that $\JUPcs$ proves an explicit form of the Ramsey axiom 
\[
  \Box(C \to A)   \leftrightarrow    [C]\Box A
\]  
from Dynamic Doxastic Logic.
\begin{lemma}
Let the formula $[C]s:(C \to A)$ be up-independent. Then
$\JUPcs$ proves 
\begin{equation}\label{eq:ramsey:1}
s:(C \to A) \ \leftrightarrow\ [C] s\cdot_C \up(C) : A.
\end{equation}
\end{lemma}
\begin{proof}
First observe that by $(\Up)$, we have $[C]\up(C):C$.

Further, since $[C]s:(C \to A)$ is up-independent, we find by $(\Indep)$ that
\[
s:(C \to A)   \leftrightarrow  [C]s:(C \to A).
\]
Finally we obtain \eqref{eq:ramsey:1} using Lemma~\ref{l:aux:1}.
\end{proof}

Often, completeness of public announcement logics is established by showing that each formula with announcements is equivalent to an announcement-free formula. Unfortunately, this approach cannot be employed for $\JUPcs$ although \eqref{eq:ramsey:1} provides a reduction property for certain formulas of the form $[C]t:A$. The reason is the hyperintensionality of justification logic~\cite{artemovFittingBook,ExploringSM}, i.e.~justification logic is not closed under substitution of equivalent formulas. Because of this, the proof by reduction cannot be carried through in $\JUPcs$, see the discussion in~\cite{BucKuzRenSacStu10LogKCA}.

\section{Conclusion}

We have introduced the justification logic $\JUP$ for subset models with belief expansion. 
We have established basic properties of the deductive system and shown its soundness. 
We have also investigated persitence properties for first-order and higher-order beliefs.
 
The next step is, of course, to obtain a completeness result for subset models with updates. We suspect, however, that the current axiomatization of $\JUP$ is not strong enough. The proof of Lemma~\ref{l:higher:1} shows that persistence of higher-order beliefs fails in the presence of a negative occurence of an $\up$-term. Thus we believe that we need a more subtle version of axiom $(\Indep)$ that distinguishes between positive and negative occurences of terms. Introducing polarities for term occurences, like in Fitting's realization procedure~\cite{Fit05APAL}, may help to obtain a complete axiomatization.



\begin{thebibliography}{10}
	\providecommand{\url}[1]{{#1}}
	\providecommand{\urlprefix}{URL }
	\expandafter\ifx\csname urlstyle\endcsname\relax
	\providecommand{\doi}[1]{DOI~\discretionary{}{}{}#1}\else
	\providecommand{\doi}{DOI~\discretionary{}{}{}\begingroup
		\urlstyle{rm}\Url}\fi
	
	\bibitem{Artemov1995OperationalModalLogic}
	Artemov, S.N.: Operational modal logic.
	\newblock Tech. Rep. MSI 95--29, Cornell University (1995)
	
	\bibitem{Artemov2001ExplicitProvability}
	Artemov, S.N.: Explicit provability and constructive semantics.
	\newblock Bulletin of Symbolic Logic \textbf{7}(1), 1--36 (2001)
	
	\bibitem{Art06TCS}
	Artemov, S.N.: Justified common knowledge.
	\newblock TCS \textbf{357}(1--3), 4--22 (2006).
	\newblock \doi{10.1016/j.tcs.2006.03.009}
	
	\bibitem{Art08RSL}
	Artemov, S.N.: The logic of justification.
	\newblock RSL \textbf{1}(4), 477--513 (2008).
	\newblock \doi{10.1017/S1755020308090060}
	
	\bibitem{Art12SLnonote}
	Artemov, S.N.: The ontology of justifications in the logical setting.
	\newblock Studia Logica \textbf{100}(1--2), 17--30 (2012).
	\newblock \doi{10.1007/s11225-012-9387-x}
	
	\bibitem{artemov2016onAggregatingPE}
	Artemov, S.N.: On aggregating probabilistic evidence.
	\newblock In: S.~Artemov, A.~Nerode (eds.) LFCS 2016, pp. 27--42. Springer
	(2016)
	
	\bibitem{ArtFit11SEP}
	Artemov, S.N., Fitting, M.: Justification logic.
	\newblock In: E.N. Zalta (ed.) The {S}tanford {E}ncyclopedia of {P}hilosophy,
	fall 2012 edn. (2012).
	\newblock
	\urlprefix\url{http://plato.stanford.edu/archives/fall2012/entries/logic-justification/}
	
	\bibitem{artemovFittingBook}
	Artemov, S.N., Fitting, M.: Justification Logic: Reasoning with Reasons.
	\newblock Cambridge University Press (2019)
	
	\bibitem{Baltag201449}
	Baltag, A., Renne, B., Smets, S.: The logic of justified belief, explicit
	knowledge, and conclusive evidence.
	\newblock APAL \textbf{165}(1), 49--81 (2014).
	\newblock \doi{http://dx.doi.org/10.1016/j.apal.2013.07.005}
	
	\bibitem{BucKuzRenSacStu10LogKCA}
	Bucheli, S., Kuznets, R., Renne, B., Sack, J., Studer, T.: Justified belief
	change.
	\newblock In: X.~Arrazola, M.~Ponte (eds.) LogKCA-10, pp. 135--155. University
	of the Basque Country Press (2010)
	
	\bibitem{BucKuzStu11JANCL}
	Bucheli, S., Kuznets, R., Studer, T.: Justifications for common knowledge.
	\newblock Applied Non-Classical Logics \textbf{21}(1), 35--60 (2011).
	\newblock \doi{10.3166/JANCL.21.35-60}
	
	\bibitem{BucKuzStu11WoLLIC}
	Bucheli, S., Kuznets, R., Studer, T.: Partial realization in dynamic
	justification logic.
	\newblock In: L.D. Beklemishev, R.~de~Queiroz (eds.) {L}ogic, {L}anguage,
	{I}nformation and {C}omputation, 18th~International {W}orkshop,
	{WoLLIC~2011}, {P}hiladelphia, {PA}, {USA}, {M}ay 18--20, 2011, Proceedings,
	\emph{LNAI}, vol. 6642, pp. 35--51. Springer (2011).
	\newblock \doi{10.1007/978-3-642-20920-8_9}
	
	\bibitem{BucKuzStu14Realizing}
	Bucheli, S., Kuznets, R., Studer, T.: Realizing public announcements by
	justifications.
	\newblock Journal of Computer and System Sciences \textbf{80}(6), 1046--1066
	(2014).
	\newblock \doi{http://dx.doi.org/10.1016/j.jcss.2014.04.001}
	
	\bibitem{Fit05APAL}
	Fitting, M.: The logic of proofs, semantically.
	\newblock APAL \textbf{132}(1), 1--25 (2005).
	\newblock \doi{10.1016/j.apal.2004.04.009}
	
	\bibitem{komaogst}
	Kokkinis, I., Maksimovi\'c, P., Ognjanovi\'c, Z., Studer, T.: First steps
	towards probabilistic justification logic.
	\newblock Logic Journal of IGPL \textbf{23}(4), 662--687 (2015).
	\newblock \doi{10.1093/jigpal/jzv025}
	
	\bibitem{KuzStu12AiML}
	Kuznets, R., Studer, T.: Justifications, ontology, and conservativity.
	\newblock In: T.~Bolander, T.~Bra{\"u}ner, S.~Ghilardi, L.~Moss (eds.) Advances
	in Modal Logic, Volume 9, pp. 437--458. College Publications (2012)
	
	\bibitem{KuzStu13LFCS}
	Kuznets, R., Studer, T.: Update as evidence: Belief expansion.
	\newblock In: S.N. Artemov, A.~Nerode (eds.) {LFCS}~2013, \emph{LNCS}, vol.
	7734, pp. 266--279. Springer (2013).
	\newblock \doi{10.1007/978-3-642-35722-0_19}
	
	\bibitem{KSweak}
	Kuznets, R., Studer, T.: Weak arithmetical interpretations for the logic of
	proofs.
	\newblock Logic Journal of IGPL \textbf{24}(3), 424--440 (2016)
	
	\bibitem{justificationLogic}
	Kuznets, R., Studer, T.: Logics of Proofs and Justifications.
	\newblock College Publications (2019)
	
	\bibitem{StuderLehmannSubsetModel2019}
	Lehmann, E., Studer, T.: Subset models for justification logic.
	\newblock In: WoLLIC 2019. Springer (2019)
	
	\bibitem{ExploringSM}
	Lehmann, E., Studer, T.: Exploring subset models for justification logic
	(submitted)
	
	\bibitem{Mkr97LFCS}
	Mkrtychev, A.: Models for the logic of proofs.
	\newblock In: S.~Adian, A.~Nerode (eds.) Logical {F}oundations of {C}omputer
	{S}cience, 4th International Symposium, {LFCS}'97, {Y}aroslavl, {R}ussia,
	{J}uly 6--12, 1997, Proceedings, \emph{LNCS}, vol. 1234, pp. 266--275.
	Springer (1997).
	\newblock \doi{10.1007/3-540-63045-7_27}
	
	\bibitem{Ren12Synthesenonote}
	Renne, B.: Multi-agent justification logic: communication and evidence
	elimination.
	\newblock Synthese \textbf{185}(S1), 43--82 (2012).
	\newblock \doi{10.1007/s11229-011-9968-7}
	
	\bibitem{Stu13JSL}
	Studer, T.: Decidability for some justification logics with negative
	introspection.
	\newblock JSL \textbf{78}(2), 388--402 (2013).
	\newblock \doi{10.2178/jsl.7802030}
	
\end{thebibliography}

\end{document}